\documentclass[10pt]{article}

\usepackage{geometry}
\usepackage{lmodern}
\usepackage[T1]{fontenc}
\usepackage[mathletters]{ucs} 
\usepackage[utf8]{inputenc}

\usepackage{graphicx}
\graphicspath{{images/}}

\usepackage{amsmath, amsthm}
\usepackage{amssymb, amsbsy}
\usepackage[dvipsnames]{xcolor}
\usepackage{hyperref}
\hypersetup{
colorlinks=true,
linkcolor=Brown,
citecolor=OliveGreen,
urlcolor=RoyalBlue,
}

\usepackage{floatflt, wrapfig}
\usepackage{xcolor}

\usepackage{graphicx}
\graphicspath{ {images/} }
\usepackage{epstopdf}
\usepackage{color, import}
\usepackage{sectsty}
\usepackage{enumitem}
\usepackage{lineno}

\usepackage{cite,, url}
\usepackage[edges]{forest}

\usepackage[linesnumbered,vlined,ruled]{algorithm2e}
\SetAlFnt{\footnotesize}
\SetAlCapFnt{\small}
\SetAlCapNameFnt{\small}

\usepackage{setspace}
\usepackage[list=true]{subcaption}
\usepackage[title]{appendix}
\usepackage{multirow}
\usepackage{array}

\usepackage{etoolbox}
\makeatletter
\patchcmd{\@maketitle}{\begin{center}}{\begin{flushleft}}{}{}
\patchcmd{\@maketitle}{\begin{tabular}[t]{c}}{\begin{tabular}[t]{@{}l}}{}{}
\patchcmd{\@maketitle}{\end{center}}{\end{flushleft}}{}{}
\makeatother

\renewenvironment{abstract}
{\small\section*
{\bfseries\noindent{\raisebox{-.15\baselineskip}{\normalsize\abstractname}}\hrulefill} 
}

\newtheorem{theorem}{Theorem}
\newtheorem{lemma}{Lemma}

\begin{document}
 \pagenumbering{gobble}

\title{Memory Optimal Dispersion by Anonymous Mobile Robots}
    \author{
    Archak Das\\
    \small \emph{Department of Mathematics, Jadavpur University, India}\\
    \small \emph{archakdas.math.rs@jadavpuruniversity.in}\\\\
    Kaustav Bose\\
    \small \emph{Department of Mathematics, Jadavpur University, India}\\
    \small \emph{kaustavbose.rs@jadavpuruniversity.in}\\\\
    Buddhadeb Sau\\
    \small \emph{Department of Mathematics, Jadavpur University, India}\\
    \small \emph{buddhadeb.sau@jadavpuruniversity.in}\\\\
    }
    \date{}
    
  \maketitle

\begin{abstract}

Consider a team of $k \leq n$ autonomous mobile robots initially placed at a node of an arbitrary graph $G$ with $n$ nodes. The \emph{dispersion} problem asks for a distributed algorithm that allows the robots to reach a configuration in which each robot is at a distinct node of the graph. If the robots are anonymous, i.e., they do not have any unique identifiers, then the problem is not solvable by any deterministic algorithm.  However, the problem can be solved even by anonymous robots if each robot is given access to a fair coin which they can use to generate random bits. In this setting, it is known that the robots require $\Omega(\log{\Delta})$ bits of memory to achieve dispersion, where $\Delta$ is the maximum degree of $G$. On the other hand, the best known memory upper bound is $min \{\Delta, max\{\log{\Delta}, \log{D}\}\}$ ($D$ = diameter of $G$), which can be $\omega(\log{\Delta})$, depending on the values of $\Delta$ and $D$. In this paper, we close this gap by presenting an optimal algorithm requiring $O(\log{\Delta})$ bits of memory. 
\end{abstract}

\section{Introduction}

\subsection{Background and Motivation}

A considerable amount of research has been devoted in recent years to the study of distributed algorithms for autonomous multi-robot system. A multi-robot system consists of a set of autonomous mobile computational entities, called $robots$, that coordinate with each other to achieve some well defined goals, such as forming a given pattern, exploration of unknown environments etc. The robots may be operating on continuous space or graph-like environments. The most fundamental tasks in graphs are \textsc{Gathering} \cite{dieudonne2016anonymous,bose2018optimal,di2017optimal,d2014gathering,klasing2008gathering,klasing2010taking,di2020gathering,czyzowicz2012meet,ta2014deterministic,kamei2019gathering,kowalski2008meet,izumi2010mobile,miller2016time,miller2015fast} and \textsc{Exploration} \cite{dereniowski2015fast,brass2011multirobot,cohen2008label,diks2004tree,duncan2006optimal,panaite1999exploring,das2018collaborative}. A relatively new problem which has attracted a lot of interest recently is \textsc{Dispersion}, introduced by Augustine and Moses Jr. \cite{AugustineM18}. The problem asks $k \leq n$ robots, initially placed arbitrarily at the nodes of an $n$-node anonymous graph, to reposition themselves to reach a configuration in which each robot is at a distinct node of the graph. The problem has many practical applications, for example, in relocating self-driven electric cars to recharge stations where finding new recharge stations is preferable to multiple cars queuing at the same station to recharge. The problem is also interesting because of its relationship to other well-studied problems such as \textsc{Exploration}, \textsc{Scattering} and \textsc{Load Balancing} \cite{AugustineM18}.

It is easy to see that the problem cannot be solved deterministically by a set of anonymous robots. Since all robots execute the same deterministic algorithm and initially they are in the same state, the co-located robots will perform the same moves. This is true for each round and hence they will always mirror each other’s move and will never do anything different. Hence, throughout the execution of the algorithm, they will stick together and as a result, dispersion cannot be achieved. Using similar arguments, it can be shown that the robots need to have $\Omega(\log{k})$ bits of memory each in order to solve the problem by any deterministic algorithm \cite{AugustineM18}. However, it has been recently shown in \cite{MollaM19} that if we consider randomized algorithms, i.e., each robot is given access to a fair coin which can be used to generate random bits, then \textsc{Dispersion} can be solved by anonymous robots with possibly $o(\log{k})$ bits of memory. In \cite{MollaM19}, two algorithms are presented for \textsc{Dispersion} from a \emph{rooted configuration}, i.e., a configuration in which all robots are situated at the same node. The first algorithm requires each robot to have $O(max\{\log{\Delta}, \log{D}\})$ bits of memory, where $\Delta$ and $D$ are respectively the maximum degree and diameter of $G$. The second algorithm requires each robot to have $O(\Delta)$ bits of memory. In \cite{MollaM19}, it is also shown that the robots require $\Omega(\log{\Delta})$ bits of memory to achieve dispersion in this setting. Notice that while the memory requirement of the second algorithm is clearly $\omega(\log{\Delta})$, that of the first algorithm too can be $\omega(\log{\Delta})$ depending on the values of $\Delta$ and $D$. In this paper, we close this gap by presenting an asymptotically optimal algorithm that requires $O(\log{\Delta})$ bits of memory.

\subsection{Related Works}

\textsc{Dispersion} was introduced in \cite{AugustineM18} where the problem was considered in specific graph structures such as paths, rings, trees as well as arbitrary graphs.  In \cite{AugustineM18}, the authors assumed $k$ $=$ $n$, i.e., the number of robots $k$ is equal to the number of nodes $n$. They proved a memory lower bound of $\Omega(\log k)$ bits at each robot and a time lower bound of $\Omega(\log D)$ rounds for any deterministic algorithm to solve the problem in a graph of diameter $D$. They then provided deterministic algorithms using $O(\log n)$ bits of memory at each robot to solve \textsc{Dispersion} on lines, rings and trees in $O(n)$ time. For rooted trees they provided an algorithm requiring $O(\Delta + \log n)$ bits of memory and $O(D^2)$ rounds and for arbitrary graphs, they provided an algorithm,  requiring $O(n\log n)$ bits of memory and $O(m)$ rounds ($m $ is the number of edges in the graph). In \cite{icdcn/KshemkalyaniA19}, a $\Omega(k)$ time lower bound was proved for $k \leq n $. In addition, three deterministic algorithms were provided in \cite{icdcn/KshemkalyaniA19} for arbitrary graphs. The first algorithm requires $O(k$log$\Delta)$ bits of memory and $O(m)$ time, ($\Delta =$ the maximum degree of the graph), the second algorithm requires $O(D\log\Delta)$ bits of memory and $O(\Delta^D)$ time, and the third algorithm requires $O($log(max$(k,\Delta)))$ bits of memory and $O(mk)$ time. Recently, a deterministic algorithm was provided in \cite{KshemkalyaniMS19} that runs in $O(\min(m,k\Delta)\log k)$ time and uses $O(\log n)$ bits of memory at each robot. In \cite{walcom/KshemkalyaniMS20}, the problem was studied on grid graphs. The authors presented two deterministic algorithms on anonymous grid graphs that achieve simultaneously optimal bounds with respect to both time and memory complexity. For the first algorithm, the authors considered the local communication model where a robot can only communicate with other robots that are present at the same node. Their second algorithm works in global communication model where a robot can communicate with other robots present anywhere on the graph. In the local communication model, they showed that the problem can be solved in an $n$-node square grid graph in  $O($min$(k,\sqrt{n}))$ time with $O(\log k)$ bits of memory at each robot. In the global communication model, the authors showed that it can be solved in $O(\sqrt{k})$ time with $O(\log k)$ bits of memory at each robot. In \cite{KshemkalyaniMS20}, the authors extended the work in global communication model to arbitrary graphs. They gave three deterministic algorithms, two for arbitrary graphs and one for trees. For arbitrary graphs, their first algorithm is based on DFS traversal and has time complexity of $O(\min(m,k\Delta))$ and memory complexity of $\Theta(\log(\max(k,\Delta)))$. The second algorithm is based on BFS traversal and has  time complexity $O(\max(D,k)\Delta(D+\Delta))$ and memory complexity $O(\max(D, \Delta\log k))$. The third algorithm in arbitrary trees is a BFS based algorithm that has time and memory complexity $O(D\max(D,k))$ and $O(\max(D, \Delta\log k))$ respectively. In \cite{agarwalla2018deterministic}, the problem was studied on dynamic rings.  Fault-tolerant  \textsc{Dispersion} was considered for the first time in \cite{mollabyz} where the authors studied the problem on a ring in presence of Byzantine robots. In \cite{MollaM19}, randomization was used to break the $\Omega(\log k)$ memory lower bound for  deterministic algorithms. In particular, the authors considered anonymous robots that can generate random bits and gave two deterministic algorithms that achieve dispersion from rooted configurations on an arbitrary graph. The memory complexity of the algorithms are respectively $O(max\{\log{\Delta}, \log{D}\})$ and $O(\Delta)$. For arbitrary initial configurations, they gave a random walk based algorithm that requires $O(\log{\Delta})$ bits of memory, but the robots do not terminate.

\subsection{Our Results}

We study \textsc{Dispersion} from a rooted configuration on arbitrary graphs by a set of anonymous robots with random bits. In \cite{MollaM19}, two algorithms with memory complexity $O(max\{\log{\Delta}, \log{D}\})$ and $O(\Delta)$ were reported. The question of whether the problem can be solved with $O(\log\Delta)$ bits of memory at each robot was left as an open problem. In this paper, we answer this question affirmatively by presenting an algorithm with memory complexity $O(\log\Delta)$. The lower bound result presented in \cite{MollaM19} implies that the algorithm is asymptotically optimal with respect to memory complexity. 

\subsection{Organization of the Paper}
In Section \ref{sec: preli}, we describe the model and introduce notations that will be used in the paper. In Section \ref{sec: algo}, we describe the main algorithm. In Section \ref{sec: correct}, we prove the correctness of our algorithm and establish the time and memory complexity. 


\section{Technical Preliminaries}\label{sec: preli}

\paragraph{Graph.} We consider a connected undirected graph $G$ of $n$ nodes, $m$ edges, diameter $D$ and maximum degree $\Delta$. For any node $v$, its degree is denoted by $\delta(v)$ or simply $\delta$ when there is no ambiguity. The nodes are anonymous, i.e., they do not have any labels. For every edge connected to a node, the node has a corresponding port number for the edge.  For every node, the edges incident to the node are uniquely identified by port numbers in the range $[0,\delta-1 ]$. There is no relation between the two port numbers of an edge.  If $u, v$ are two adjacent nodes then port$(u,v)$ denotes the port at $u$ that corresponds to the edge between $u$ and $v$.

\paragraph{Robots.} Robots are anonymous, i.e., they do not have unique identifiers. Each robot has $O(\log{\Delta})$ bits of space or memory for computation and to store information. Each robot has a fair coin which they can use to generate random bits. Each robot can communicate with other robots present at the same node by message passing:  a robot can broadcast some message which is received by all robots present at the same node. The size of a message is no more than its memory size because it cannot generate a message whose size is greater than its memory size. Therefore, the size of a message must be $O(\log{\Delta})$. Also, when there are many robots (co-located at a node) broadcasting their messages, it is not possible for a robot to receive all of these messages due to limited memory. When there is not enough memory to receive all the messages, it receives only a subset of the messages. The view of a robot is local: the only things that a robot can `see' when it is at some node, are the edges incident to it. The robots have access to the port numbers of these edges. It cannot `see' the other robots that may be present at the same node. The only way it can detect the presence of other robots is by receiving messages that those robots may broadcast. The robots can move from one node to an adjacent node. Any number of robots are allowed to move via an edge. When a robot moves from a node $u$ to node $v$, it is aware of the port through which it enter $v$. 


\paragraph{Time Cycle.} We assume a fully synchronous system. The time progresses in rounds. Each robot knows when a current round ends and when the next round starts. Each round consists of the following.

\begin{itemize}
 \item The robots first performs a series of synchronous computations and communications. These are called \emph{subrounds}. In each subround, a robot performs some local computations and then broadcasts some messages. The messages received in the $i$th subround are read in the $(i+1)$th subround.  The local computations are based on its memory  contents (which contains the messages that it might have received in the last subround and other information that it had stored) and a random bit generated by the fair coin. 
 
 \item Then robots move through some port or remains at the current node. 
\end{itemize}


\paragraph{Problem Definition.} A team of $k$ ($\leq n$) robots are initially at the same node of the graph $G$. The \textsc{Dispersion} problem requires the robots to re-position themselves so that i) there is at most one robot at each node, and ii) all robots have terminated within a finite number of rounds.

\section{The Algorithm}\label{sec: algo}

\subsection{Local Leader Election}

Before presenting our main algorithm, we give a brief description of the \textsc{LeaderElection()} subroutine. We adopt this subroutine from \cite{MollaM19}. When $k \geq 1$ robots are co-located together at a node, \textsc{LeaderElection()} subroutine allows exactly one robot to be selected as the leader within one round. Formally, 1) if $k = 1$, the robot finds out that it is the only robot at the node, 2) if $k > 1$, after finitely many rounds (with high probability), i) exactly one robot is elected as leader, ii) all robots can detect when the process is completed. Each robot starts off as a candidate for leader. In the first subround, every robot broadcasts `start'. If a robot finds that it has received no message, it then concludes that it is the only robot at the node. Otherwise, it concludes that there are multiple robots at the node and does the following. In each subsequent subround, each candidate flips a fair coin. If heads, it broadcasts `heads', otherwise it does not broadcast anything. If a robot gets tails, and receives at least one (`heads') message, it stops being a candidate. This process is repeated until exactly one robot, say $r$, broadcasts in a given subround. In this subround, $r$ broadcasts `heads', but receives no message, while all other non-candidate robots have not broadcasted, but received exactly one message. So $r$ elects itself as the leader, and all robots detect that the process is completed. The process requires $O(1)$ bits of memory at each robot and terminates in $O(\log k)$ subrounds with high probability.

\subsection{Overview of the Algorithm}

In this subsection, we present a brief overview of the algorithm. The execution of our algorithm can be divided into three stages. In the first stage, the robots, together as a group, perform a DFS traversal in the search of empty nodes, starting from the node where they are placed together initially. We shall call this node the $root$ and denote it by $v_R$. Whenever the group reaches an empty node, they perform the \textsc{LeaderElection()} subroutine to elect a leader. The leader settles at that node, while the rest of the group continues the DFS traversal. Note that the settled robot does not terminate. This is because  when the robots that are performing the DFS return to that node, they need to detect that the node is occupied by a settled robot. Recall that a robot cannot distinguish between an empty node and a node with a terminated robot. Therefore, the active settled robot helps the travelling robots to distinguish between an occupied node and an empty node, and also provides them with other information that are required to correctly execute the DFS. The size of the travelling group decreases by one, each time the DFS traversal reaches an empty node. The first stage completes when each robot has found an empty node for itself. Let $r_L$ denote the last robot that finds an empty node, $v_L$, for itself. Although dispersion is achieved, this robot will not terminate. The other settled robots do not know that dispersion is achieved and will remain active. Therefore $r_L$ needs to revisit those nodes and ask the settled robots to terminate. First $r_L$ will return to the root $v_R$ via the \emph{rootpath} which is the unique path in the DFS tree from $v_L$ to $v_R$.  This is the second stage of the algorithm. In the third stage, $r_L$ performs a second DFS traversal and asks the active settled robots to terminate. Since the active settled robots play crucial in the DFS traversal, $r_L$ needs to be careful about the order in which it should ask the settled robots to terminate. Finally, $r_L$ terminates after it returns to $v_L$.

A pseudocode description of the algorithm is given in Algorithm \ref{algo:main}. In Table \ref{table}, we give details of the variables used by the robots. If $variable\_name$ is some variable, then we shall denote the value of the variable stored by $r$ as $r.variable\_name$.

\begin{table}
\begin{center}
    \begin{tabular}{  l | p{11cm} }
    \hline
    \textbf{Variable} & \textbf{Description} \\ \hline
    
    \emph{role} & It indicates the role that the robot is playing in the algorithm. It takes values from $\{$\texttt{explore}, \texttt{settled}, \texttt{return}, \texttt{acknowledge}, \texttt{done}$\}$. Initially, $role \leftarrow \texttt{explore}$.  \\ \hline
    
    \emph{entered} & It indicates the port through which the robot has entered the current node. Initially, $entered \leftarrow \emptyset$. For simplicity, assume that it is automatically updated when the robot entered a node.   \\ \hline
    
    \emph{received} & It indicates the message(s) received by the robot in the current subround. After the end of each subround, the messages are erased, i.e., it is reset to $\emptyset$. Initially, $received \leftarrow \emptyset$. \\ \hline
    
    \emph{direction} & It indicates the direction of movement of a robot during a DFS traversal. It takes values from $\{$\texttt{forward}, \texttt{backward}$\}$. Initially, $role \leftarrow \texttt{forward}$. \\ \hline
    
    \emph{parent} & For a settled robot on some node, it indicates the port number towards the parent of that node in the DFS tree.  Initially, $parent \leftarrow \emptyset$. \\ \hline
    
    \emph{child} & For a settled robot on some node that is on the rootpath, it indicates the port number towards the child of that node in the DFS tree that is on the rootpath.  Initially, $child \leftarrow \emptyset$. \\ \hline
    
    \emph{visited} & For a settled robot on some node, it indicates whether the node where the robot is settled has been visited by $r_L$ in the third stage.  Initially, $visited \leftarrow 0$. \\ \hline
    
    \end{tabular}
\end{center}
\caption{Description of the variables used by the robots}\label{table}
\end{table}

\begin{algorithm}[]
    \setstretch{0.01}
    \SetKwInOut{Input}{Input}
    \SetKwInOut{Output}{Output}
    \SetKwProg{Fn}{Function}{}{}
    \SetKwProg{Pr}{Procedure}{}{}
    \Pr{\textsc{Dispersion()}}{

    $r \leftarrow$ myself
    
    \uIf{$r.role =$ \texttt{settled}}{

      \uIf{I am queried}{
    
	  \textsc{Broadcast}($role = settled, parent = r.parent, child = r.child, visited = r.visited$)
	
      }

      \uElseIf{$r.received =$ ``Set child $x$''}{
      
	$r.child \leftarrow x$
      
      }
 	  \uElseIf{$r.received =$ ``Set visited = $1$''}{
      
	$r.visited \leftarrow 1$
      
      }
      \ElseIf{$r.received = ``terminate$'' }{
 	  \textsc{Terminate}()
 	  }

	}

      \uElseIf{$r.role =$ \texttt{explore}}{ 
      
	\textsc{Query()}
    
	\uIf{$r.received = \emptyset$}{
      
	  \textsc{LeaderElection()}\\
	  \uIf{I am alone}{
	    $r.role \leftarrow$ \texttt{return}\\
	    Move through $r.entered$\\
	  }
	  \uElseIf{I am elected as leader}{
	    $r.role \leftarrow$ \texttt{settled}\\
	    $r.parent \leftarrow r.entered$ 
	  }
	  \ElseIf{I am not elected as leader}{
	  
	    \uIf{$r.entered = \emptyset$}{Move via port 0}
	    \Else{
	          \If{($r.entered +1 = r.entered$) mod $\delta$}
	               {$r.direction \leftarrow$ \texttt{backward} }
	    Move via port $(r.entered + 1)$mod$\delta$
	    
	    }
	  }
	
	}
	\ElseIf{$r.received = ``role = settled, parent = x, child = \emptyset, visited = 0$''}{
      
	  \uIf{$r.direction =$ \texttt{forward}}{
	
	    $r.direction \leftarrow$ \texttt{backward}\\
	    Move via port $r.entered$
	  
	    }
	  \ElseIf{$r.direction =$ \texttt{backward}}{

	      \uIf{$(r.entered + 1)$mod$\delta = x$}{
	      
	      Move via port $(r.entered + 1)$mod$\delta$
	  
	      }
	      \Else{
	  
	      $r.direction \leftarrow$ \texttt{forward}\\
	      Move via port $(r.entered + 1)$mod$\delta$
	  
	      }
	    
	    }
	
      
	}
      
      }
      
      \uElseIf{$r.role =$ \texttt{return}}{
      
	\textsc{Query}()\\
	\If{$r.received = ``role = settled, parent = x, child = \emptyset, visited = 0$''}{
	  \uIf{$x \neq \emptyset$}{
	  
	    \textsc{Broadcast}(``Set child $r.entered$'')\\
	    Move via port $x$\\
	  
	    }
	  \Else{
	  
	    \textsc{Broadcast}(``Set child $r.entered$'')\label{code: 1}\\
	    $r.direction \leftarrow$ \texttt{forward} \\
	    $r.role \leftarrow$ \texttt{acknowledge}\\
	    $r.entered \leftarrow \emptyset$ \label{code: 2} \\
	  
	}
	
	}
      
      }
      
      \uElseIf{$r.role =$ \texttt{acknowledge}}{
      
	\textsc{Query}()\\
	\uIf{$r.received = ``role = settled, parent = x, child = y, visited = 0$''}{
	
	  \textsc{Broadcast}(``Set visited = 1'')\\
	  \uIf{$r.entered = \emptyset$}{Move via port 0}
	    \Else{
	      
	      \uIf{$r.entered = (r.entered + 1)$mod$\delta$}{
		$r.direction \leftarrow$ \texttt{backward}\\
		\textsc{Broadcast}(``terminate'')
		}
	    \ElseIf{$y = (r.entered + 1)$mod$\delta$}{
	      \textsc{Broadcast}(``terminate'')}
	    
	      Move via port $(r.entered + 1)$mod$\delta$}

	}
	\uElseIf{$r.received = ``role = settled, parent = x, child = y, visited = 1$''}{
	
	  \uIf{$r.direction =$ \texttt{forward}}{
	  
	     $r.direction \leftarrow$ \texttt{backward}\\
	     Move via port $r.entered$
	  
	  }
	  \ElseIf{$r.direction =$ \texttt{backward}}{
	  
	   \uIf{$x = (r.entered + 1)$mod$\delta$}{
	       
	       \textsc{Broadcast}(``terminate'')
	       }
	   \uElseIf{$y = (r.entered + 1)$mod$\delta$}{
	           $r.direction \leftarrow$ \texttt{forward}\\
	           \textsc{Broadcast}(``terminate'')}
	   \Else{$r.direction \leftarrow$ \texttt{forward}}        
	  Move via port ($r.entered + 1$) mod $\delta$
	  }
	
	}

     \ElseIf{$r.received = \emptyset$}{
     \uIf{$r.direction =$ \texttt{forward}}{
     $r.direction \leftarrow$ \texttt{backward}\\
     
     }
     \ElseIf{$r.direction =$ \texttt{backward}}{
     $r.role =$ \texttt{done}
     }
     Move via port $r.entered$
     }
     } 
      \ElseIf{$r.role =$ \texttt{done}}{\textsc{Terminate}()}
 }

\caption{Dispersion}
    \label{algo:main} 
\end{algorithm}

\subsection{Detailed Description of the Algorithm}

In the starting configuration, all robots are present at the root node $v_R$. Initially, \emph{role} of each robot is \texttt{explore}. In the first stage, the robots have to perform a DFS traversal together as a group. This group of robots is called the \emph{exploring group}. Whenever the exploring group reaches an empty node (a node with no settled robot), one of the robots will settle at that node, i.e., it will change its \emph{role} to \texttt{settled} and remain at that node. For the rest of the algorithm, it does not move. However, it stays active and checks for any received messages. A settled robot can receive three types of messages: 

\begin{itemize}
 \item it may receive a query about the contents of its internal memory
 \item it may receive an to instruction to change the value of some variable
 \item it may be asked to terminate
\end{itemize}

When queried about its memory, it broadcasts a message containing its \emph{role}, \emph{parent}, \emph{child} and \emph{visited}. If it is asked to change the value of some variable or terminate, then it does so accordingly. Any robot with role  \texttt{explore}, \texttt{return} or \texttt{acknowledge}, in the first subround of any round, broadcasts a message querying about internal memory of any settled robot at the node. If it receives no message in the second subround, then it concludes that there is no settled robot at that node. Whenever the robots find that there is no settled robot at the node, during the first stage, they start the \textsc{LeaderElection()} subroutine to elect a leader. For any robot $r$, \textsc{LeaderElection()} results in one of the following outcomes:

\begin{itemize}
 \item it is elected as the leader
 \item it is not elected as the leader
 \item it finds that it is the only robot at that node
\end{itemize}

In the first case, it changes $r.role$ to \texttt{settled} and sets $r.parent$ equal to $r.entered$. Recall that $r.entered$ is the port through which it entered the current node and in the beginning, $r.entered$ is set to $\emptyset$. We shall call  $r.parent$ the \emph{parent port} of the node where $r$ resides. We shall refer to a robot that has set its $role$ to \texttt{settled} as a \emph{settled robot}. In the second case, it will continue the DFS: if $r.entered = \emptyset$, it leaves via port 0 and if $r.entered \neq \emptyset$, it leaves via port $(r.entered + 1)$ mod $\delta$. If $(r.entered + 1) = r.entered$ mod $\delta$, it changes its variable \emph{direction} to \texttt{backward} before exiting the node. Recall that the variable \emph{direction} is used to indicate the direction of the movement during a DFS traversal. In the third case, it changes $r.role$ to \texttt{return}. 

Now consider the case where the robots find that there is a settled robot at the node. If the \emph{direction} is set to \texttt{forward} when they encounter the settled robot, it indicates the onset of a cycle. So the robots change the \emph{direction} to \texttt{backward} and leave the node via the port through which they entered it. Now suppose that the \emph{direction} is set to \texttt{backward} when they encounter the settled robot. Recall that the robots have received from the settled robot, say $a$, a message which contains $a.parent$. The robots check if $a.parent$ is equal to the port number through which it entered, say $z$,  plus 1 (modulo the degree of the node). If yes, it implies that the robots have moved through all edges adjacent to the node, and hence they leave the node via $a.parent$ which is the port through which they  entered for the first time. If no, then it means that they have not moved through the port $(z + 1)$mod$\delta$ before. So they change the direction to \texttt{forward} and leave via $(z + 1)$mod$\delta$. 

The DFS traversal in the first stage ends when a robot, say $r_L$, with \emph{role} set to \texttt{explore}, finds that it is the only robot at a node, say $v_L$. Recall that when this happens, $r_L$ changes its \emph{role} to \texttt{return}. At this point, the first stage ends, and the second stage starts. It then leaves $v_L$ via the port through which it entered. In each of the following rounds where the \emph{role} of $r_L$ is \texttt{return}, it does the following. In the first subround, it broadcasts a query. In the next subround, it receives a message from the settled robot at that node which contains its \emph{parent}. If the obtained value of \emph{parent}, say $x$, is not $\emptyset$, it means that $r_L$ is yet to reach the root $v_R$. Then $r_L$ broadcasts an instruction for the settled robot to change the value of its \emph{child} to the port via which $r_L$ entered the node. This value of $child$ will be called the \emph{child port} of the node. After broadcasting the instruction, $r_L$ leaves through the port $x$. If $x = \emptyset$, then it means that $r_L$ has reached the root $v_R$. In this case, $r_L$ broadcasts the same instruction and then changes the values of $r_L.role$, $r_L.direction$ and $r_L.entered$ to respectively \texttt{acknowledge}, \texttt{forward} and $\emptyset$ . At this point the second stage ends, and the third stage starts.

In the following rounds, $r_L$ with \emph{role} \texttt{acknowledge} does the following.  In the first subround, it broadcasts a query. It either receives a reply or does not. If it receives a message, then it contains the values of \emph{parent}, say $x$, and \emph{child}, say $y$, and \emph{visited} of the settled robot at that node. Now, the value of variable \emph{visited} can be $0$ or $1$. If the value of \emph{visited} is $0$, it denotes that the settled robot is visited for the first time in the third stage. The robot $r_L$ then broadcasts a message instructing the settled robot to change the value of its variable $visited$ to $1$. Now $(r_L.entered + 1)$mod$\delta$ can be equal to $r_L.entered$ (the case of one degree node) or $y$ or neither of them. In the former case, it changes its variable $direction$ to \texttt{backward}. In the first two cases, it broadcasts a message instructing the settled robot to terminate and leaves through port   $(r_L.entered + 1)$mod$\delta$. If $(r_L.entered + 1)$mod$\delta$ is neither equal to $r_L.entered$, nor equal to $y$, $r_L$ just exits through $(r_L.entered + 1)$mod$\delta$ without broadcasting any message for termination. If the value of \emph{visited} is $1$, it denotes that the settled robot has been visited before in the third stage. If the value of variable $direction$ of $r_L$ is \texttt{forward}, it changes the value of $direction$ to \texttt{backward} and exits through the port through which it entered the node at the previous round. Otherwise the value of variable $direction$ of $r_L$ is \texttt{backward}. In this case, three sub-cases arise.  If $(r_L.entered + 1)$mod$\delta$ is equal to $x$, then $r_L$ broadcasts a message instructing the settled robot to terminate, and then $r_L$ exits through the port  $(r_L.entered + 1)$mod$\delta$. Otherwise if, $(r_L.entered + 1)$mod$\delta$ is equal to $y$, then also $r_L$ broadcasts a message instructing the settled robot to terminate, changes the variable $direction$ to \texttt{forward} and then $r_L$ exits through the port  $(r_L.entered + 1)$mod$\delta$. If   $(r_L.entered + 1)$mod$\delta$ is neither equal to $x$, nor $y$, then $r_L$ changes $direction$ to  \texttt{forward} and exits through  $(r_L.entered + 1)$mod$\delta$. Now, we consider the case where $r_L$ in third stage does not receive any answer to its query. If its $direction$ is set to \texttt{forward}, it changes its $direction$ to \texttt{backward} and then exits through the same port by which it entered the  node in the previous round. If its direction is set to \texttt{backward}, then it means that $r_L$ was at $v_L$ in the previous round. So $r_L$ changes its \emph{role} to \texttt{done} and leaves the node through the port via which it entered. Then it will reach $v_L$ in the next round and  it will find that its \emph{role} is \texttt{done} and terminate.

\section{Correctness Proof and Complexity Analysis}\label{sec: correct}

The first stage of our algorithm is the same as that of \cite{MollaM19}. The robots simply perform a DFS traversal. Whenever a new node is visited, one of the robots settle there. The DFS continues until $k$ distinct nodes are visited. To see that the DFS traversal can be correctly executed in our setting, it suffices to verify that the robots can correctly ascertain 1) if a node is previously visited and 2) if all neighbors of a node have been visited. For 1), observe that the presence of settled robot at a node indicates that the node has already been visited. So, when the robots with \emph{direction} \texttt{forward} go to a node which has a settled robot, it backtracks, i.e., it changes its $direction$ to \texttt{backward} and leaves the node via the port through which it had entered. For 2), observe that the port $p$ through which robots first enters a node $v$ is set as its parent port, i.e., the robot settled at $v$ sets its variable $parent$ to $p$. Then the robots will move through all other ports with \emph{direction} \texttt{forward} in the order $p+1, p+2, \ldots, \delta-1, 0, 1, \ldots , p-1$ (unless the DFS is stopped midway for $k$ distinct nodes have been visited). This is because if the robots leaves via a port $q$ (with \emph{direction} \texttt{forward}), it re-enters $v$ via the same port $q$ after some rounds (with \emph{direction} \texttt{backward}) and then leaves via $(q+1)$mod$\delta(v)$ (with \emph{direction} \texttt{forward}) in the next round if $(q+1)$mod$\delta(v) \neq p$. Clearly, when $(q+1)$mod$\delta(v) = p$, it indicates that the robots have moved through all ports other than $p$ with \emph{direction} \texttt{forward}, i.e., all neighbors of $v$ have been visited. The robots can check if $(q+1)$mod$\delta(v) = p$ because their variable \emph{entered} is equal to $q$ and the variable \emph{parent} of the robot settled at $v$ is equal to $p$. Inspecting the pseudocode of Algorithm \ref{algo:main}, it is easy to see that these are correctly implemented in the algorithm. Therefore, have the following result.

\begin{theorem}\label{thy 1}
There is a round $t_1$, at the beginning of which 
\begin{enumerate}
 \item each node of $G$ has at most one robot
 \item \emph{role} of exactly one robot $r_L$ is \texttt{explore} and the role of the remaining $k-1$ robots is \texttt{settled}
 \item if $V' \subseteq V$ is the set of nodes occupied by robots, then $G[V']$ (the subgraph of $G$ induced by $V'$) is connected
 \item if $E' \subseteq E$ is the set of edges corresponding to the variable \emph{parent} of robots in $\mathcal{R} \setminus \{r_L\}$ and variable \emph{entered} of $r_L$, then the graph $T = T(V',E')$ is a DFS spanning tree of $G[V']$,
 \item $r_L$ is at a leaf node $v_L$ of $T$.
\end{enumerate}

\end{theorem}

In the following lemmas, we present some observations regarding the execution of DFS traversal in the first stage. 

\begin{lemma}\label{observations1}

     If $v$ is a non-rootpath node, then the exploring group leaves it via its parent port once. If $v$ is a rootpath node other than $v_L$, then the exploring group leaves it via its child port once.

\end{lemma}

\begin{lemma}\label{observations2}

     Suppose that $v$ is a non-rootpath node and the exploring group leaves $v$ through its parent port at round $t$ with $direction$ \texttt{backward}. If the  exploring group returns to $v$ at some round $t', t < t' \leq t_1$, then its $direction$ must be \texttt{forward}.

\end{lemma}

\begin{lemma}\label{observations3}

     Suppose that $v$ is a rootpath node and the exploring group leaves $v$ through its child port at round $t$ with $direction$ \texttt{forward}. If the  exploring group returns to  $v$ at some round $t', t < t' \leq t_1$, then its $direction$ must be \texttt{forward}.

\end{lemma}

There is a unique path, i.e., the rootpath $v_R = v_1,v_2,\ldots,v_s = v_L$ in $T(V',E')$ from $v_R$ to $v_L$. Furthermore, for any consecutive vertices $v_i, v_{i+1}$ on the path 1) if $i+1 < s$, the variable \emph{parent} of the settled robot at $v_{i+1}$ is set to port$(v_{i+1},v_i)$ and 2) if $i+1 = s$, the variable \emph{entered} of robot $r_L$ at $v_{i+1}$ is set to port$(v_{i+1},v_i)$. So, according to our algorithm, $r_L$ will move along this path to reach $v_R$. For each node $v_i, i<s$, on the rootpath, when $r_L$ reaches $v_i$ along its way to $v_R$, it instructs the settled robot at $v_i$ to set its variable \emph{child} to port$(v_i,v_{i+1})$. Therefore, we have the following result.

\begin{theorem}\label{th 2}
There is a round $t_2$, at the beginning of which 
\begin{enumerate}
 \item $r_L$ is at $v_R$ with $r_L.role =$ \texttt{return}
 \item each node of $T(V',E') \setminus \{v_L\}$ has a settled robot
 \item if $v_R = v_1,v_2,\ldots,v_s = v_L$ is the rootpath and $r_{i}$ is the settled robot at $v_i, i<s$, then $r_i.child =$ port$(v_i,v_{i+1})$
 \item if $r$ is a settled robot on a non-rootpath node, then $r.child = \emptyset$.
\end{enumerate}

\end{theorem}

From round $t_2 +1$, $r_L$ will start a second DFS traversal. This DFS traversal is trickier than the earlier one because the settled robots will one by one terminate during the process. Recall that the settled robots played important role in the first DFS. We shall prove that $r_L$ will correctly execute the second DFS traversal. In fact, we shall prove that the DFS traversal in the first stage is exactly same as the DFS traversal in the third stage in the sense that if the exploring group is at node $v$ at round $i < t_1$ (in the first stage), then $r_L$ is at node $v$ at round $t_2 + i$ (in the third stage).  

Let us first introduce a definition. In the following definition, whenever we say `at round', it is to be understood as `at the beginning of round'. Round $i$ in the first stage is said to be \emph{identical} to round $j$ in the third stage if the exploring group at round $i$ and $r_L$ at round $j$ are at the same node, say $u$ and one of the following holds:\\

\begin{description}

    \item[I1] At round $i$, there is no settled robot at $u$, the exploring group contains more than one robot and the variable $direction$ for each robot in the exploring group is set to \texttt{forward}. At round $j$ there is a settled robot at $u$ with its variable $visited$ set to $0$. The variables $direction$ and $entered$ of $r_L$ at round $j$  are equal to those of each robot in the exploring group at round $i$.

    \item[I2] At round $i$, there is a settled robot at $u$, the variable $direction$ for each robot in the exploring group is set to \texttt{forward} and the variable $entered$ for each robot in the exploring group is $\neq \emptyset$. At round $j$, either there is a terminated robot at $u$, or there is a settled robot at $u$ with its variable $visited$ set to $1$. The variables $direction$ and $entered$ of $r_L$ at round $j$  are equal to those of each robot in the exploring group at round $i$ .
    
    \item[I3] At round $i$, there is a settled robot at $u$, the variable $direction$ for each robot in the exploring group is set to \texttt{backward} and the variable $entered$ for each robot in the exploring group is $\neq \emptyset$. At round $j$, there is an active settled robot at $u$ with its variable $visited$ set to $1$. The variables $direction$ and $entered$ of $r_L$ at round $j$ are equal to those of each robot in the exploring group at round $i$  .
    
\end{description}

\begin{lemma}\label{identical}
Round $i$  is identical to $t_2 + i$ for all $1 \leq i < t_1$. 
\end{lemma}

\begin{proof}
We prove this by induction. At the beginning of round $1$, the exploring group (consisting of more than one robots) is at the root node $v_R$, there is no settled robot at $v_R$, and the variables $direction$ and $entered$ of each robot in the exploring group are set to \texttt{forward} and $\emptyset$ respectively. Note that when the robot $r_L$ enters the root node $v_R$ at round $t_2$, it sets its $direction$ and $entered$ to \texttt{forward} and $\emptyset$ respectively, and does not move (See line \ref{code: 1}-\ref{code: 2} in Algorithm \ref{algo:main}). Hence at the beginning of round $t_2 +1$, $r_L$ is at node $v_R$, with its $direction$ and $entered$ set to \texttt{forward} and $\emptyset$ respectively. Furthermore, there is a settled robot at $v_R$ at round $t_2 + 1$. This is the robot which was elected as the leader in round $1$ and had settled there. Since the variable $visited$ of this robot was initially set to $0$ and it has been not instructed to change it thus far, it is still set to $0$ at the beginning of $t_2 + 1$. These observations imply that I1 holds and round $1$  is identical to $t_2 + 1$.  

Now assume that round $j$ is identical to round $t_2 + j$ for all $1 \leq j \leq i-1, i<t_1$. We shall prove that round $i$ is identical to round $t_2 + i$. Let the robot $r_L$ be at node $u$ at round $t_2 +i-1$. Then it implies that the exploring group was at node $u$ at round $i-1$.  Now we can have following three cases.

    \textbf{Case 1.} Suppose that at round $i-1$, 1) there is no settled robot at $u$, 2) the exploring group contains more than one robot and 3) the variable $direction$ for each robot in the exploring group is set to \texttt{forward}. Suppose that their variable $entered$ are set to $p$. Then the robots will perform the \textsc{LeaderElection()} protocol and one of them will settle. The remaining robots will move via 0 if $p = \emptyset$ or otherwise, will move via $(p + 1)$ mod $\delta$. In the later case, the robots will change their $direction$ to \texttt{backward} iff $(p + 1)$ mod $\delta = p$. Since we assumed that round $i-1$ is identical to round $t_2 + i - 1$, at the beginning of round $t_2 + i - 1$, 1) $r_L$ is at $u$, 2) there is a settled robot at $u$ with its variable $visited$ set to $0$, and 3) the variables $direction$ and $entered$ of $r_L$ are equal to \texttt{forward} and $p$. According to our algorithm, $r_L$ will move via 0 if $p =\emptyset$ or otherwise, will move via $(p + 1)$mod$\delta$. In the later case, $r_L$ will change its $direction$ to \texttt{backward} iff $(p + 1)$mod$\delta = p \Leftrightarrow \delta(u) = 1$. Hence the exploring group at round $i$ and $r_L$ at $t_2 +i$ must be at the same node, say $v$, and have same values of $direction$ and $entered$. So now it remains to show that one of I1, I2 or I3 is true for round $i$ and $t_2+i$. For this, consider the following two cases.

    \textbf{Case 1a.} First consider the case where the exploring group and $r_L$ exit $u$ with $direction$ \texttt{forward} at round $i-1$ and $t_2 + i - 1$ respectively. At the beginning of round $i$, the exploring group has more than one robots as $i<t_1$. Now, at the beginning of round $i$, $v$ either has no settled robot or has a settled robot. In the first case, it implies that exploring group is entering $v$ for the first time at round $i$. The robots will then perform the \textsc{LeaderElection()} protocol and one of them, say $a$, will settle. By the induction hypothesis, it implies that $r_L$ will enter $v$ for the first time in the third stage at round $t_2 + i$. Hence the settled robot $a$ must be there with the value of $visited$ set to $0$. So I1  holds and hence round $i$ is identical to round $t_2 + i$. In the second case, it implies that the exploring group had entered $v$ at some previous round $j < i$. Then by the induction hypothesis, $r_L$ had visited $v$ at round $t_2 + j < t_2 + i$. Hence, if the robot $a$ is active, then variable $visited$ of $a$ is $1$, or otherwise $a$ has terminated. So I2  holds and hence round $i$ is identical to round $t_2 + i$.  
    
    
    \textbf{Case 1b.} Now consider the case where the exploring group and $r_L$ exit $u$ with $direction$ \texttt{backward} at round $i-1$ and $t_2 + i - 1$ respectively. This implies that $u$ is a one degree node. Also, the exploring group and $r_L$ were at $v$ at round $i-2$ and $t_2 + i - 2$ respectively. Therefore, there must be a settled robot, say $b$, present at $v$ , when the exploring group visits it at round $i$, as the node was visited earlier. So when $r_L$ enters $v$ at round $t_2+i$, if $b$ is still active, its $visited$ value is set to $1$ and I3 holds. We prove that this is the only case. For the sake of contradiction, let us assume that $b$ has terminated. Consider the following cases.
    
    \begin{itemize}
        \item First let $v$ be a non-rootpath node. Since $b$ has terminated, it implies that $r_L$ had exited $v$ via its parent port with $direction$ \texttt{backward} at some round $t_2 + l < t_2 + i-1$. Then by the induction hypothesis, the exploring group exited $v$ via its parent port with $direction$ \texttt{backward} at some round $l < i-1$. But then the fact  that the exploring group returns to $v$ with $direction$ \texttt{backward} at round $i$ contradicts Lemma \ref{observations2}.
        
        \item Now let $v$ be a rootpath node. Since $b$ has terminated, it implies that $r_L$ had exited $v$ via its child port with $direction$ \texttt{forward} at some round $t_2 + l < t_2 + i-1$. Then by the induction hypothesis, the exploring group exited $v$ via its child port with $direction$ \texttt{forward} at some round $l < i-1$.  But then the fact  that the exploring group returns to $v$ with $direction$ \texttt{backward} at round $i$ contradicts Lemma \ref{observations3}.
    \end{itemize}

     \textbf{Case 2.} At the beginning of round $i-1$, 1) there is a settled robot at $u$, 2) the variable $direction$ for each robot in the exploring group is set to \texttt{forward} and 3) the variable $entered$ for each robot in the exploring group is $p \neq \emptyset$. According to our algorithm, the exploring group will change their $direction$ to \texttt{backward} and leave the node via $p$. Since we assumed that round $i-1$ is identical to round $t_2 + i - 1$, at the beginning of round $t_2 + i - 1$, $r_L$ is at $u$, there is either an active settled robot at $u$ with its variable $visited$ set to $1$ or a terminated robot, and the variables $direction$ and $entered$ of $r_L$ are equal to \texttt{forward} and $p$ respectively. According to our algorithm, $r_L$ will change it $direction$ to \texttt{backward} and leave the node via $p$. Hence the exploring group at round $i$ and $r_L$ at $t_2 +i$ must be at the same node, say $v$, both having $direction$ set to \texttt{backward} and $entered$ set to port$(v,u)$. 
     
     It is clear that the exploring group and $r_L$ was at $v$ at rounds $i-2$ and $t_2 +i -2$ respectively. Hence, there is a settled robot at $v$, say $a$, at round $i$, as it had been visited at least once before. Also, $a$ is still situated at $v$ at round $t_2 +i$, either active or terminated. If it is active, then the value of its variable $visited$ is $1$, as $r_L$ had been at $v$ at round $t_2 +i-2$. So, in this case I3 holds. Using the same arguments as in Case 1b, we can show that this is the only possible case.

  \textbf{Case 3.} At the beginning of round $i-1$, 1) there is a settled robot at $u$, say $a$, 2) the variable $direction$ for each robot in the exploring group is set to \texttt{backward} and 3) the variable $entered$ for each robot in the exploring group is $p \neq \emptyset$. Let the $parent$ of $a$ be equal to $x$. According to our algorithm, the robots will not change their $direction$ if $(p+1)$mod$\delta = x$ (Case 3a), and otherwise, it will change it to \texttt{forward} (Case 3b). In any case, they will leave the node via port $(p+1)$mod$\delta$. Since we assumed that round $i-1$ is identical to round $t_2 + i - 1$, at the beginning of round $t_2 + i - 1$, 1) $r_L$ is at $u$, 2) there is a settled robot at $u$ with its variable $visited$ set to $1$, and 3) the variables $direction$ and $entered$ of $r_L$ are equal to \texttt{backward} and $p$ respectively. Since a settled robot does not move or change its $parent$, the settled robot at $u$ at round $t_2 + i - 1$ is $a$ and its $parent$ is set to $x$. According to our algorithm, $r_L$ will not change its $direction$ if $(p+1)$mod$\delta = x$ (Case 3a), and otherwise, it will change it to \texttt{forward} (Case 3b). In any case, they will leave the node via port $(p+1)$mod$\delta$. Hence at round $i$ and $t_2 +i$, the exploring group and $r_L$ must be at the same node, say $v$, with same $direction$ and $entered$. So now it remains to show that one of I1, I2 or I3 is true for round $i$ and $t_2+i$.

    \textbf{Case 3a.} Clearly in this case $v$ must have been visited at least once before round $i$ in the first stage. Hence a robot, say $c$, is already settled there at round $i$. Hence, at round $t_2 +i$,  either $c$ is active with $visited$ set to $1$, or  is terminated. In the first case, I3 holds and hence we are done. We can prove that the later case is impossible using the same arguments as in Case 1b.

       \textbf{Case 3b.} If there is no settled robot at $v$ at round $i$, then the exploring group is visiting $v$ for the first time and one of them, say $b$, will settle there. It implies from our induction hypothesis that $r_L$ is also visiting $v$ for the first time in the third stage at round $t_2+i$. Hence it will find $b$  with its $visited$ set to $0$. So I1 holds. If there is a settled robot at $v$ at round $i$, say $c$, then at round $t_2 + i$, $v$ has $c$ which is either terminated or its variable $visited$ set to $1$. So we see that I2 holds.
       \end{proof}

\begin{theorem}\label{term1}
By round $t_2 + t_1$ all the settled robots have terminated.
\end{theorem}

\begin{proof}
 A settled robot at a non-rootpath node will terminate if $r_L$ moves from that node to its parent and a settled robot at a rootpath node will terminate if $r_L$ moves from that node to its child. Recall that if $v$ is a non-rootpath node, then in Phase 1 the exploring group leaves it via its parent port once. Also, if $v$ is a rootpath node other than $v_L$, then in Phase 1 the exploring group leaves it via its child port once. So the result follows from Lemma \ref{identical}. 
 \end{proof}

\begin{theorem}
At round $t_2 + t_1 + 2$, $r_L$ terminates at $v_L$.
\end{theorem}

\begin{proof}
It follows from Lemma \ref{identical} that $r_L$ will be at $v_{s-1}$ at round $t_2 + t_1-1$. It is easy to see that it will then move to $v_s = v_L$ with $direction$ \texttt{forward}. So at round $t_2 + t_1$, $r_L$ is at $v_L$ with $direction$ set to \texttt{forward}. Since $v_L$ has no settled robot, $r_L$ will change its $direction$ to  \texttt{backward} and exit through the port through which it entered $v_L$. But moving through this port leads to $v_{s-1}$. The settled robot at this node is already terminated by Lemma \ref{term1}. Hence at round $t_2 + t_1 +1$, $r_L$ enters with $direction$ \texttt{backward} a node where it does not receive any message. So it will then change its $role$ to \texttt{done} and exit the port through which it entered the node. Therefore, at round $t_2 + t_1 +2$, $r_l$ again returns to $v_L$, this time with $role$ \texttt{done}, and hence will terminate. 
\end{proof}

\begin{theorem}
Algorithm \ref{algo:main} is correct.
\end{theorem}

\begin{proof}
It follows from the above results that at round $t_2 + t_1 + 2$, dispersion accomplished and all robots have terminated. Hence Algorithm \ref{algo:main} solves \textsc{Dispersion} from any rooted configuration.  
\end{proof}

\begin{theorem}\label{}
Algorithm \ref{algo:main} requires $O(\log \Delta)$ bits of memory at each robot and this is optimal in terms of memory complexity. 
\end{theorem}

\begin{proof}
The \textsc{LeaderElection}() subroutine costs $O(1)$ bits of memory for each robot. Among the variables, $role$, $visited$ and $direction$ costs $O(1)$ bits of memory, and the variables $entered$, $parent$, $child$, $received$ costs $O(\log \Delta)$ bits of memory for each robot. Hence the algorithm requires $O(\log \Delta)$ bits of memory at each robot. The optimality follows from the lower bound result proved in \cite{MollaM19}.
\end{proof}

\begin{theorem}\label{}
Algorithm \ref{algo:main} requires $\Theta(k^2)$ rounds in the worst case (assuming that \textsc{LeaderElection()} terminates each time).  
\end{theorem}

\begin{proof}
 Exactly $k$ distinct nodes are visited in our algorithm. In the first stage, movement of the exploring group takes $O(m')$ rounds where $m'$ is the number of edges in the subgraph of $G$ induced by these $k$ vertices. Clearly $m'= O(k^2)$. So the first stage requires $O(k^2)$ rounds. The third stage requires $(k^2)$ rounds as  well since apart from the last two rounds, it is exactly identical to the first stage. Clearly the second round takes $O(k)$ rounds. So the overall round complexity is $O(k^2)$.
 
 To see that our analysis is tight, we show an instance where $\Omega(k^2)$ rounds will be required. Consider the graph of size $k$ in Fig. \ref{fig: time}. Here port$(v_R,v_1) = 0$, port$(v_1,v_R) = 0$, port$(v_1,v_2) = 1$, port$(v_1,v_L) = 2$. Here, after reaching $v_1$, the exploring group will go $v_2$. It is easy to see that $\Theta(k^2)$ rounds will be spent inside the $(k-3)-$clique. Finally the last robot will return to $v_1$ and then move to $v_L$.    
 \end{proof}

\begin{figure}[htb!]
\centering
\fontsize{10pt}{10pt}\selectfont
\def\svgwidth{0.55\textwidth}
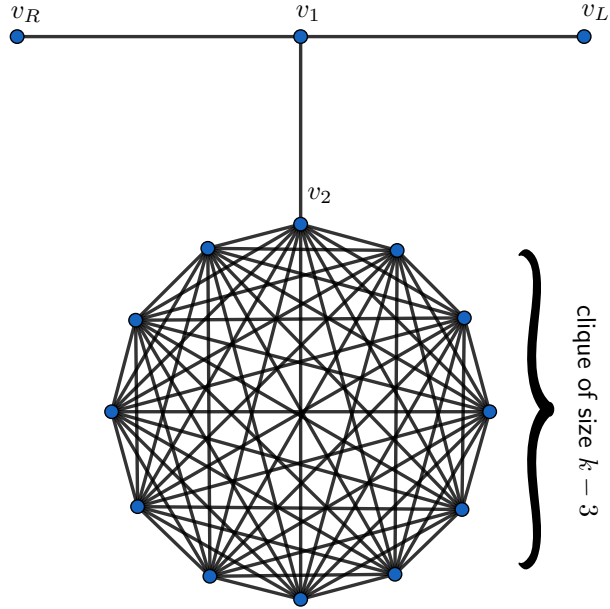
\caption{An example where Algorithm \ref{algo:main} requires $\Theta(k^2)$ rounds to complete.}
\label{fig: time}
\end{figure}

\section{Concluding Remarks}\label{sec: conclu}

We have presented a memory optimal randomized algorithm for \textsc{Dispersion} from rooted configuration by anonymous robots. This resolves an open problem posed in \cite{MollaM19}. Time complexity of our algorithm is $\Theta(k^2)$ rounds in the worst case (assuming that \textsc{LeaderElection()} terminates each time). Any algorithm that solves the problem requires $\Omega(k)$ rounds in the worst case.
To see this, consider a path with $n \geq k$ nodes with all robots initially at one of its one degree nodes. An interesting open problem is to close this gap.

For arbitrary configuration, the random walk based algorithm presented in \cite{MollaM19} requires the robots to stay active indefinitely. Therefore an interesting open question is whether it is possible to solve the problem by anonymous robots from non-rooted configurations without requiring robots to stay active indefinitely.

\paragraph{Acknowledgement.}
We would like to thank Pritam Goswami for valuable discussions.
The first two authors are supported by UGC, Govt. of India, and NBHM DAE, Govt. of India respectively. We would like to thank the anonymous reviewers for their valuable
comments which helped us to improve the quality and presentation of the paper.

\bibliographystyle{plainurl}
\bibliography{random_dispersion}

\end{document}